\documentclass[reqno,a4paper,oneside,11pt]{amsart}

\usepackage[usenames,dvipsnames]{xcolor}
\usepackage[colorlinks=true,linkcolor=Blue,citecolor=Green]{hyperref}
\usepackage{amsmath,amssymb,epsf}
\usepackage{amsfonts,amsbsy,amscd,stmaryrd, wasysym}
  \SetSymbolFont{stmry}{bold}{U}{stmry}{m}{n}
\usepackage{latexsym,amsopn}
\usepackage{amsthm}
\usepackage{esint}
\usepackage{mathrsfs}
\usepackage{slashed} 
\usepackage{enumitem}

\usepackage[margin=1.3in,footskip=0.25in]{geometry}
\allowdisplaybreaks

\usepackage[utf8]{inputenc}
\usepackage[T1]{fontenc}   

\usepackage{graphicx}
\usepackage{multirow}
\usepackage{float}

\newcounter{smallarabics}
\newenvironment{arabicenumerate}
{\begin{list}{{\normalfont\textrm{(\arabic{smallarabics})}}}
  {\usecounter{smallarabics}\setlength{\itemindent}{0cm}
   \setlength{\leftmargin}{5ex}\setlength{\labelwidth}{4ex}
   \setlength{\topsep}{0.75\parsep}\setlength{\partopsep}{0ex}
   \setlength{\itemsep}{0ex}}}
{\end{list}}

\newcounter{smallroman}

\newcommand{\ben}{\begin{arabicenumerate}}  
\newcommand{\een}{\end{arabicenumerate}}

\newtheorem{theorem}{Theorem}[section]

\newtheorem{proposition}[theorem]{Proposition}
\newtheorem{lemma}[theorem]{Lemma}

\theoremstyle{definition}
\newtheorem{definition}[theorem]{Definition}
\newtheorem{remark}[theorem]{Remark}
\newtheorem{example}[theorem]{Example}

\newcommand{\beq}{\begin{equation}}
\newcommand{\eeq}{\end{equation}}
\newcommand{\bea}{\begin{aligned}}
\newcommand{\eea}{\end{aligned}}
\newcommand{\bear}{\begin{array}{rl}}
\newcommand{\eear}{\end{array}}
\newcommand{\bex}{\begin{example}}
\newcommand{\eex}{\end{example}}
\def\bel{\begin{lemma}}
\def\eel{\end{lemma}}
\def\bet{\begin{theoreme}}
\def\eet{\end{theoreme}}
\def\bed{\begin{definition}}
\def\eed{\end{definition}}
\def\ber{\begin{remark}}
\def\eer{\end{remark}}
\def\bep{\begin{proposition}}
\def\eep{\end{proposition}}
\newcommand{\qeds}{\qed\medskip}

\let\origmaketitle\maketitle
\def\maketitle{
  \begingroup
  \def\uppercasenonmath##1{} 
  \let\MakeUppercase\relax 
	\origmaketitle
  \endgroup
}

\def\bar{\overline}

\usepackage{calrsfs}
\DeclareMathAlphabet{\pazocal}{OMS}{zplm}{m}{n}
\def\cA{{\pazocal A}}

\def\cC{{\pazocal C}}
\def\cD{{\pazocal D}}

\def\cF{{\pazocal F}}

\def\cM{{\pazocal M}}
\def\cN{{\pazocal N}}
\def\cO{{\pazocal O}}

\def\cQ{{\pazocal Q}}

\def\cU{{\pazocal U}}

\def\sH{\mathcal{H}}

\def\rr{{\mathbb R}}

\def\cc{{\mathbb C}}
\def\nn{{\mathbb N}}

\def\ss{{\mathbb S}}

\def\sgn{{\rm sgn}}

\DeclareMathOperator{\tr}{tr}
\DeclareMathOperator{\supp}{supp}

\DeclareMathOperator{\WF}{WF}

\def\14{\frac{1}{4}}

\def\12{\frac{1}{2}}

\def\d{{\rm d}}

\def\sgn{{\rm sgn}}

\DeclareSymbolFont{boldoperators}{OT1}{cmr}{bx}{n}
\SetSymbolFont{boldoperators}{bold}{OT1}{cmr}{bx}{n}

\makeatletter
\newcommand*{\defeq}{\mathrel{\rlap{%
                     \raisebox{0.34ex}{$\m@th\cdot$}}%
                     \raisebox{-0.4ex}{$\m@th\cdot$}}%
                     =}
                     
\makeatother
\makeatletter
\newcommand*{\eqdef}{=\mathrel{\rlap{%
                     \raisebox{0.34ex}{$\m@th\cdot$}}%
                     \raisebox{-0.4ex}{$\m@th\cdot$}}%
                     }
\makeatother

\def\WF{{\rm WF}}

\def\dVol{\mathop{}\!d{\rm vol}}

\def\MI{{\rm M}_{\rm I}}
\def\MII{{\rm M}_{\rm II}}
\def\MIII{{\rm M}_{\rm III}}

\def\Mm{{\rm M}_{-}}
\def\M+{{\rm M}_{+}}
\def\Mc{{\rm M}_{c}}

\def\2Sol{{\rm Sol}_{{\rm L}^{2}}}

\newcommand{\abs}[1]{{\left\vert #1 \right\vert}}
\newcommand{\norm}[1]{{\left\Vert #1 \right\Vert}}
\newcommand{\IP}[1]{{\left\langle #1 \right\rangle}}

\numberwithin{equation}{section}

\begin{document}

\title[Quantum fields in rotating black holes]{\Large Exploring quantum fields in rotating black holes}

\author{Christiane K. M. \textsc{Klein}}
\address{Universit\'e Grenoble Alpes, Institut Fourier, 100 rue des Maths, 38610 Gi\`eres, France\\ CY Cergy Paris Universit\'e, 2 avenue Adolphe Chauvin, 95302 Cergy-Pontoise, France\\
Department of Mathematics, University of York, Heslington, York YO10 5DD, United Kingdom}
\email{christiane.klein@york.ac.uk}
\keywords{Quantum Field Theory on curved spacetimes, Hadamard states, Kerr-de Sitter spacetime, Hawking temperature}
\subjclass[2010]{81T13, 81T20, 35S05, 35S35}
\thanks{\emph{Acknowledgments.} The author would like to thank D.W. Janssen for helpful comments on an earlier version of the manuscript. The author acknowledges support from the Agence Nationale de la Recherche (ANR), funding ANR-20-CE40-0018-01. The author was funded by the Deutsche Forschungsgemeinschaft (DFG, German
Research Foundation) – Projektnummer 531357976.}
\begin{abstract}
In this paper, we discuss the Unruh state for a free scalar quantum field on Kerr-de Sitter under the assumption of mode stability. We summarise the proof of its Hadamard property that was previously given in 
[C.Klein, Annales Henri Poincaré 24 (2023) 7, 2401-2442]
for sufficiently small black-hole rotation and cosmological constant, and show how it can be generalised to any subextreme black-hole angular momentum in the same range of the cosmological constant. This is done by extending a geometric analysis of the trapped set of the Kerr spacetime
[D. Häfner, C. Klein, Lett.Math.Phys. 114 (2024) 5, 119]
to Kerr-de Sitter. Moreover, we discuss the application of this state in the numerical study of quantum effects at the inner horizon
[C. Klein, M. Soltani, M. Casals, S. Hollands, Phys.Rev.Lett. 132 (2024) 12, 121501]
, and describe a universality result for these effects 
[P. Hintz, C. Klein, Class.Quant.Grav. 41 (2024) 7, 075006].
\end{abstract}
\maketitle

\section{Introduction}

In recent years, rotating black holes have become an important playing field to test General Relativity (GR), as well as its extensions and modifications. This is due to the possibilities of observing the shadows of nearby supermassive black holes \cite{EHT} or the gravitational waves emitted by the merger of binary black holes or black hole-neutron star systems \cite{ LIGO}. However, black holes also pose interesting theoretical and conceptual problems, such as the black hole information loss paradox or the strong cosmic censorship conjecture.

In contrast to the information loss paradox, which is linked to the semi-classical effect of black hole evaporation, the strong cosmic censorship conjecture can be studied in classical GR. It is related to the inner horizons present in rotating and charged black holes. Beyond these inner horizons, the spacetimes can be smoothly continued. However, the continuations are not uniquely fixed by any set of complete initial data for the Einstein equations sufficient to specify the black hole spacetimes up to their inner horizon. In this sense, the inner horizons of charged and rotating black holes constitute examples of Cauchy horizons. Penrose's strong cosmic censorship conjecture states that the presence of a Cauchy horizon is an unstable feature and that any generic perturbation of the black hole's initial data will render the spacetime inextendible across the inner horizon at a certain regularity. Since the notion of generic initial data and the required regularity still need to be specified, this is a whole family of conjectures rather than a single one.

In view of the difficulty of studying the Einstein equations with generic initial data, a first step towards (dis-)proving any particular formulation of the conjecture is to study the scalar wave equation on the spacetimes of interest. Within this toy-model, it was shown that the $H^1_{loc}$-version of the conjecture can be violated classically in charged black holes in asymptotic de-Sitter space \cite{HV, CCDHJ}. However, considering the scalar as a quantum field can remedy the situation, since the quantum stress-energy tensor of the scalar field generically shows a stronger divergence at the inner horizon than its classical counterpart \cite{HWZ,HKZ}.

 No such violation of the $H^1_{loc}$-version of the conjecture is known for rotating black holes. Nonetheless, quantum effects may have a strong influence on the geometry at the inner horizon and may, for example, alter the nature of the singularity that forms there. As a first step to test this, one can
study free scalar quantum fields in rotating black hole spacetimes.

There are different approaches to quantum field theory that could be used for such a study, out of which the algebraic approach is particularly well-suited for studying quantum fields in curved spacetimes. The approach focuses first on the algebraic structure of the observables and introduces states at a later stage.
While the algebra of observables for a free scalar field is well understood, performing explicit calculations of expectation values requires an explicitly given state for the theory on a rotating black hole spacetime. This state should also satisfy the Hadamard property, which is necessary for the definition of the stress-energy tensor and other non-linear observables \cite{HW} and, ideally, be physically motivated.
 Hadamard states are abstractly known to exist \cite{FNW, JS}, but there are many situations where an explicitly constructed, physically motivated example employable for calculations is missing. 

In this paper, we review the explicit construction of the Unruh state on Kerr-de Sitter spacetimes under the assumption of mode stability \cite{Klein}. Mode stability is believed to hold on all subextreme Kerr-de Sitter spacetimes, but so far it has only been shown for small angular momentum by perturbation around Schwarzschild-de Sitter \cite{Dyatlov:QNM} and for small cosmological constant (or, equivalently, black hole mass) by perturbation around Kerr \cite{Hintz:QNM}.

The Unruh state is not only physically motivated \cite{U}, but also satisfies the Hadamard property. We also show how the proof of the Hadamard property, given in \cite{Klein} for slowly rotating black holes with small cosmological constants, can be extended to all subextreme rotating black holes with a small cosmological constant. The argument is purely geometrical, and generalises a corresponding result obtained in \cite{KHa} for the Kerr spacetime to Kerr-de Sitter. Finally, we discuss how this state can be used to study the stress-energy tensor of the quantum field at the inner horizon of the black hole. We recall numerical results obtained using this state \cite{KSCH}. By giving a reformulation of the results in \cite{HK}, we not only show that these numerical results are, to leading order, state-independent, but also how this universality of the leading divergence can be generalised to a larger class of quadratic observables.

The paper is structured as follows. In Section \ref{sec:KdS}, we introduce the Kerr and Kerr-de Sitter spacetime and review some of its properties. In Section~\ref{sec:Unruh}, we discuss the free scalar field theory, recall the construction of the Unruh state for this theory on the Kerr-de Sitter spacetime, and show how the proof of the Hadamard property of the Unruh state can be extended. The application of the state for a discussion of the stress-energy tensor at the inner horizon is presented in Section~\ref{sec:appli}. We summarise in Section~\ref{sec:Summary}.

\section{The Kerr(-de Sitter) spacetime}
\label{sec:KdS}
The Kerr(-de Sitter) spacetimes are a family of solutions to the vacuum Einstein equations with a positive cosmological constant $\Lambda$, $Ric(g)=\Lambda g$. They describe stationary rotating black holes. In this section, we give a brief account of its most important properties for the purpose of this work.

Let $\lambda=\Lambda/3\geq 0$, $M>0$, and $a\in\rr$, and set
\begin{align*}
\Delta_r=(1-\lambda r^2)(r^2+a^2)-2Mr\, .
\end{align*}
The roots of this function describe the locations of the various horizons in the rotating black-hole spacetime. 
If $r_1$ and $r_2$ are two consecutive, distinct roots of this polynomial as a function of $r$, then we can define the block of the Kerr(-de Sitter) spacetime between the corresponding horizons in terms of the global Boyer-Lindquist coordinates $(t,r,\theta,\varphi)$ as the manifold $\rr_t\times(r_1,r_2)\times \ss^2_{\theta,\varphi}$, equipped with the metric
\begin{align}
\label{eq:metric BL}
g=\frac{\Delta_\theta a^2\sin^2\theta-\Delta_r}{\chi^2\rho^2}\d t^2+\frac{\sigma^2\sin^2\theta}{\chi^2\rho^2}\d \varphi^2+\frac{\rho^2}{\Delta_r}\d r^2+\frac{\rho^2}{\Delta_\theta}\d  \theta^2-2\frac{a\sin^2\theta\mu}{\chi^2\rho^2}\d t\d\varphi\, .
\end{align}
Here, we have defined
\begin{align*}
\rho^2&=r^2+a^2\cos^2\theta\, , & \Delta_\theta&=1+\lambda a^2\cos^2\theta\, ,  \\
 \chi&=1+\lambda a^2\, , & \mu&=(r^2+a^2)\Delta_\theta-\Delta_r=\lambda\rho^2(r^2+a^2)+2Mr\, , 
\end{align*}
and
\begin{align*}
\sigma^2&=(r^2+a^2)^2 \Delta_\theta-a^2\sin^2\theta \Delta_r =\chi(r^2+a^2)\rho^2+2Mra^2\sin^2\theta\, .
\end{align*}
We are particularly interested in the subextreme parameter regime in which $\Delta_r$ has three (resp. two for $\lambda=0$) distinct real positive roots $r_-<r_+<r_c$ (resp. $r_-<r_+$). In this case, the three relevant blocks of the Kerr(-de Sitter) spacetime are $\MI=\rr_t\times(r_+,r_c)\times\ss^2_{\theta,\varphi}$, $\MII=\rr_t\times (r_-,r_+)\times\ss^2_{\theta,\varphi}$, and $\MIII=\rr_t\times (r_c,\infty)\times\ss^2_{\theta,\varphi}$. In the case of $\lambda=0$, the third block is absent, and the interval $(r_+,r_c)$ in the definition of block $\MI$ is replaced by $(r_+,\infty)$.

The parameters $a$ and $M$ represent the angular momentum per unit mass and the mass of the black hole. By defining the new constants $\tilde a=aM^{-1}$, $\tilde{\lambda}=M^2\lambda$, changing coordinates to $\tilde{t}=M^{-1}t$ and $\tilde{r}=M^{-1}r$, and performing the conformal transformation $g\to \tilde{g}=M^{-2} g$, one can set $M=1$. In the rest of this work, we will do so and drop the tilde-decoration for ease of notation. The subextreme parameter regime in terms of the rescaled parameters $a$ and $\lambda$ is depicted in Figure~\ref{fig:para reg}.

\begin{figure}
\includegraphics[width=0.5\textwidth]{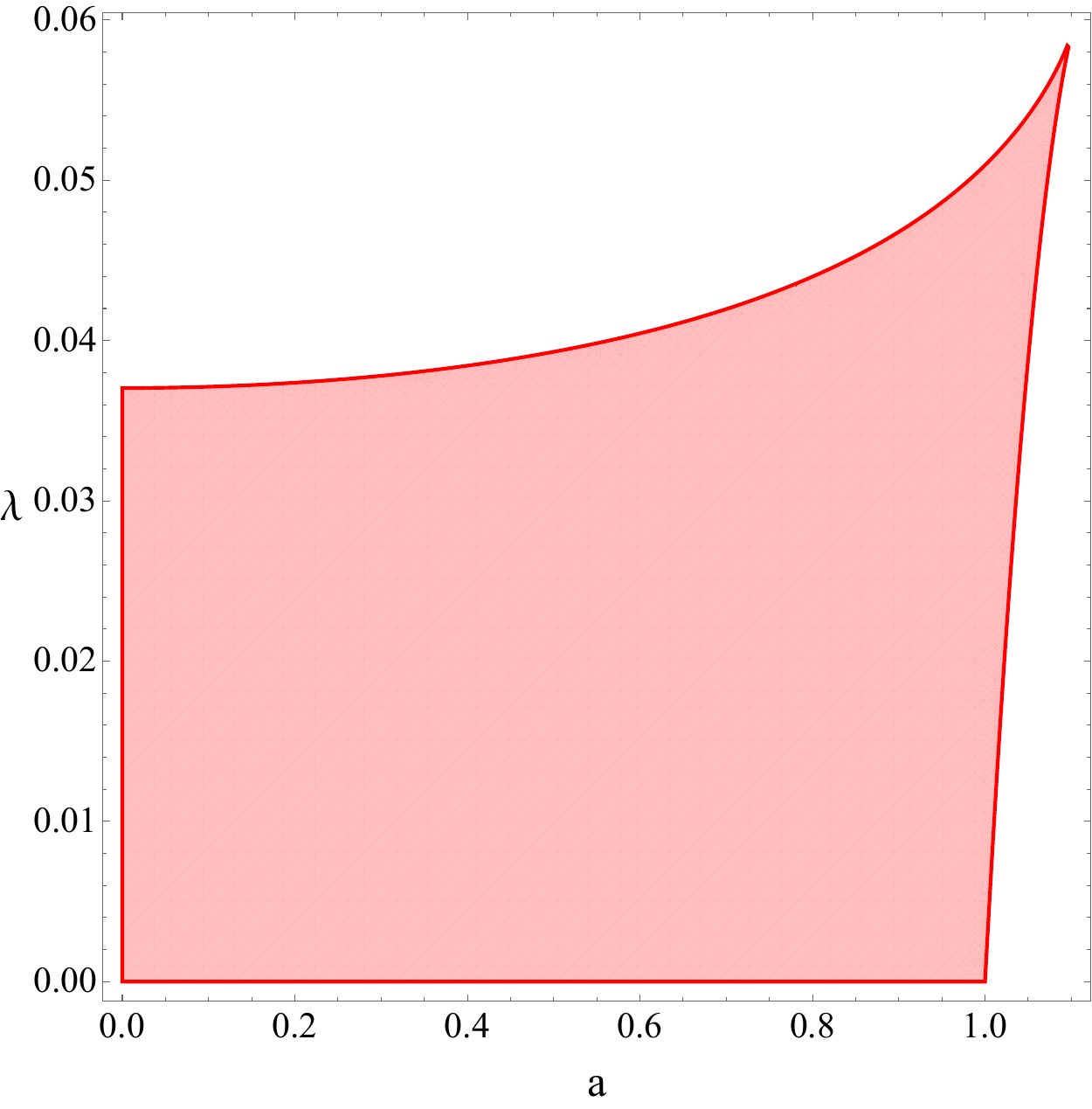}
\caption{The subextreme parameter range of Kerr-de Sitter for $M=1$ in terms of the cosmological constant $\Lambda =3\lambda$ and the black-hole angular momentum per unit mass $a$. Reproduced from \cite{thesis}(C. Klein, PhD thesis, 2023) licensed under a Creative Commons Attribution (CC BY) license.}
\label{fig:para reg}
\end{figure}

The inverse metric in Boyer-Lindquist coordinates reads \cite{Salazar}
\begin{align*}
\rho^2 g^{-1}=-\frac{\chi^2\sigma^2}{\Delta_\theta\Delta_r}\partial_t^2-\frac{2a\chi^2\mu}{\Delta_\theta\Delta_r}\partial_t\partial_\varphi+\frac{\chi^2\left[\Delta_r-a^2\sin^2\theta\Delta_\theta\right]}{\sin^2\theta\Delta_\theta\Delta_r}\partial_\varphi^2+\Delta_r\partial_r^2+\Delta_\theta\partial_\theta^2
\end{align*}
We select a time orientation for the blocks of the Kerr(-de Sitter) spacetime by demanding that
\begin{equation}
-\nabla t=-g^{-1}(\d t,\cdot)=\frac{\chi^2\sigma^2}{\rho^2\Delta_\theta\Delta_r}\left(\partial_t+\frac{a\mu}{\sigma^2}\partial_\varphi\right)
\end{equation}
is future-directed on $\MI$, $-\partial_r$ is future-directed on $\MII$, and $\partial_r$ is future-directed on $\MIII$.

For a causal covector $k$ on $\MI$, we say that it is future-directed if $k(-\nabla t)>0$, and analogously on $\MII$ and $\MIII$, $k$ is future directed if $k(-\partial_r)>0$ or $k(\partial_r)>0$, respectively.
We will denote the set of future-directed null covectors by $\cN^+$. Similarly, we define past-directed null covectors and $\cN^-$. We denote $\cN=\cN^+\cup\cN^-$.

From the form \eqref{eq:metric BL} of the metric, it is apparent that each of the blocks has two Killing vector fields, given in Boyer-Lindquist coordinates by $(\partial_t)^a$ and $(\partial_\varphi)^a$.

\subsection{The Kruskal extension}
\label{sec:Kruskal}
The Boyer-Lindquist blocks do not contain the various horizons of the Kerr(-de Sitter) spacetime, which play an important role in the construction of states by bulk-to-boundary correspondence \cite{DMP, HWZ, Klein}.

However, the blocks introduced above can be embedded isometrically into the Kruskal manifolds $\Mm =\rr_{U_-}\times\rr_{V_-}\times\ss^2_{\theta,\varphi_-}\setminus\{r(U_-, V_-)=0, \theta=\pi/2\}$, $\M+=\rr_{U_+}\times\rr_{V_+}\times\ss^2_{\theta,\varphi_+}$ or $\Mc=\{(U_c, V_c)\in \rr^2: U_cV_c<1\}\times \ss^2_{\theta,\varphi}$, equipped with the metric
\begin{align}
g_i=&f_1^{i} \left( V_i d U_i-U_i d V_i \right)^2 + f_2^{i}\left( V_i^2 d U_i^2 + U_i^2 d V_i^2 \right) + f_3^{i} d U_id V_i \\\nonumber
&+ f_4^{i} d \varphi_i \left(V_id  U_i-U_id V_i \right)+\frac{\rho^2}{\Delta_\theta} d \theta^2 + g_{\varphi\varphi}d \varphi_i^2\, ,
\end{align}
 where $i\in\{-,+,c\}$, and where
\begin{align*}
f_1^{i}=&\frac{a^2 \sin^2 \theta \Delta_\theta G_i^2(r+r_i)^2}{4\kappa_i^2\chi^2\rho^2(r_i^2+a^2)^2}\, ,\\
f_{2}^{i} =& \frac{a^{2} \sin^{2} \theta (r_i+r) G_{i}^{2} \Delta_r}{ 4 \kappa_ {i}^{2} \chi^{2} \rho^{2} (r_{i}^{2}+a^{2}) (r^{2}+a^{2}) (r-r_{i}) } \left( \frac{\rho_{i}^{2}}{r_i^2+a^2} + \frac{\rho^2}{r^2+a^2} \right)\, ,\\
f_3^{i}=& \frac{G_i\Delta_r}{2\kappa_i^2\chi^2\rho^2(r-r_i)}\left(\frac{\rho_i^4}{(r_i^2+a^2)^2} + \frac{\rho^4}{(r^2+a^2)^2} \right)\, ,\\
f_4^{i}=& s_i \frac{a \sin^2 \theta G_i}{\kappa_i\chi^2\rho^2(r_i^2+a^2)} \left[ (r^2+a^2)(r+r_i)\Delta_\theta +\frac{\rho_i^2 \Delta_r }{r-r_i}\right]\, .
\end{align*}
Here, $r(U,V)$ is implicitly defined by the equation
\begin{align*}
\frac{r-r_+}{U_+V_+}=G(r)=-e^{-2\kappa_+r}(r-r_-)^{\tfrac{r_-}{r_+}}
\end{align*}
for $\lambda=0$, and by
\begin{align}
G_i(r)=\frac{r-r_i}{U_iV_i}=-s_i\prod_{j\neq i}\abs{r-r_j}^{-\tfrac{s_i\kappa_i}{s_j\kappa_j}}\, , 
\end{align}
for $\lambda>0$, where $i,j$ run over the four roots of $\Delta_r$ and $s_i=\sgn(\partial_r\Delta_r(r_i))$. 
\begin{align*}
\kappa_i=\frac{\abs{\partial_r\Delta_r(r_i)}}{2\chi(r_i^2+a^2)}
\end{align*} 
is the surface gravity of the horizon at $r=r_i$.

Let us also define the function
\begin{align}
r_*(r)=(r+)\sum\limits_{i}\frac{s_i}{2\kappa_i}\log\abs{r-r_i}\, ,
\end{align}
where the term in the brackets is only present for $\lambda=0$.
To embed $\MI$ into $\M+$ and $\Mc$ while maintaining the time orientation, one can use the coordinate transformations
\begin{align*}
U_i=-s_ie^{-s_i\kappa_i (t-r_*)}\, , \quad V_i=s_ie^{s_i\kappa_i(t+r_*)}\, , \quad \varphi_i=\varphi-\frac{a}{r_i^2+a^2}t\,,\quad i\in\{+,c\} .
\end{align*}
In this way, $\MI$ is identified with $\M+\cap\{U_+<0,V_+>0\}$ and $\Mc\cap\{U_c>0,V_c<0\}$. In a similar way, $\MII$ can be identified with $\M+\cap\{U_+>0,V_+>0\}$ or $\Mm\cap\{U_-<0,V_-<0\}$ by setting
\begin{align*}
U_i&=s_ie^{-s_i\kappa_i (t-r_*)}\, , \quad V_i=s_ie^{s_i\kappa_i(t+r_*)}\, , \quad \varphi_i=\varphi-\frac{a}{r_i^2+a^2}t\, , \quad i\in\{-,+\} .
\end{align*}
$\MIII$ can be embedded into $\Mc\cap\{U_c>0,V_c>0\}$ via
\begin{align*}
U_c=e^{\kappa_c (t-r_*)}\, , \quad V_c=e^{-\kappa_c(t+r_*)}\, , \quad \varphi_c=\varphi-\frac{a}{r_c^2+a^2}t\, .
\end{align*}

The Kruskal extension $\M+$ contains the event horizon of the black hole, $\sH_+^R:=\{U_+=0,V_+>0\}$ and the long horizons $\sH_+:=\{V_+=0\}$. Together with the set $\{U_+=0, V_+<0\}$, these make up the bifurcate Killing horizon $\{r=r_+\}\subset \M+$, which is generated by the Killing vector field 
\begin{align}
v_+=\kappa_+\left(V_+\partial_{V_+}-U_+\partial_{U_+}\right)=\partial_t+\tfrac{a}{r_+^2+a^2}\partial_\varphi\, ,
\end{align}
where the last expression is valid in any of the Boyer-Lindquist blocks that can be embedded into $\M+$.

Similarly, for $\lambda>0$, the Kruskal extension $M_c$ encompasses the cosmological horizon $\sH_c^L:=\{U_c>0, V_c=0\}$ and the long horizon $\sH_c:=\{U_c=0\}$. They are part of the bifurcate Killing horizon $\{r=r_c\}\subset \Mc$ generated by the Killing field
\begin{align}
v_c=\kappa_c\left(U_c\partial_{U_c}-V_c\partial_{V_c}\right)=\partial_t+\tfrac{a}{r_c^2+a^2}\partial_\varphi\, .
\end{align}

The physical Kerr(-de Sitter) spacetime will be defined as
\begin{align*}
\cM=\MII\cup\sH_+^R\cup\MI\left(\cup\sH_c^L\cup\MIII\right)\, .
\end{align*}
For $\lambda=0$, it is a submanifold of $\M+$, and for $\lambda>0$ it can be isometrically embedded into 
\begin{align*}
\tilde \cM= \left(\M+\sqcup \Mc\right)/\sim\,, 
\end{align*}
 where $\sim$ is the identification of $\{U_+<0, V_+>0\}\subset \M+$ with $\{U_c>0, V_c<0\}\subset \Mc$ via the embedding of $\MI$ described above.
Considering $\cM$ as a submanifold of the maximal analytic extension of Kerr(-de Sitter), the inner  horizon $\sH_-:=\{V_-=0, U_-<0\}$ is part of the future boundary of $\cM$.\footnote{The reason we do not consider $\{U_-=0, V_-<0\}$, the other part of the future boundary, is that the former is expected to remain present in the case of gravitational collapse to the corresponding black hole. Therefore, it is considered to be of larger physical significance.}

\subsection{The null geodesics and the trapped set}
\label{sec:null geos}
 
An important part in the analysis of the singularity structure of states is played by null geodesics. Studies of the null geodesic structure can, for example, be found in \cite{oN, Chandra} for Kerr and in \cite{SZ, HLKK, Dyatlov, thesis} for Kerr-de Sitter. A vital aspect for the application in this work is the behaviour of (partially) trapped null geodesics.

An inextendible, affinely parametrised null geodesics $\gamma:(\tau_-,\tau_+)\to \cM$ is called forwards (backwards) trapped if it satisfies $\tau_\pm=\pm\infty$ and if there exists constants $r_-<r_0<R_0<\infty$ such that $r_0<r(\gamma(\tau))<R_0$ for all $\pm \tau>0$. By studying the geodesic equation, it can be shown that a forward or backward trapped geodesic in $\cM$ either lies completely in $\sH_+^R$ or $\sH_c^L$, or satisfies $r_+<r_0<R_0<r_c$ (after an affine rescaling if necessary).
 One can show that inextendible null geodesics on $\cM$ which are neither forwards nor backwards trapped must approach either $\sH_+$ or $\sH_c$ towards the past.
 
We will refer to the set of forwards or backwards trapped null geodesics for which $r_+<r_0<R_0<r_c$ as $\Gamma^\pm$, and we set $\Gamma=\Gamma^+\cup\Gamma^-$. The intersection of these two sets is the trapped set $K=\Gamma^+\cap\Gamma^-$.

The trapped set of Kerr(-de Sitter) has been studied in \cite{Dyatlov, DZ, PV1, PV2}. By relating affine geodesics to points in $T^*\cM$, $K$ can be viewed as a subset of $T^*\MI$. There exists a $\lambda_0>0$ so that for all $0\leq \lambda<\lambda_0$, $K$ can be specified in Boyer-Lindquist coordinates by \cite{Dyatlov}
\begin{equation}
\label{eq:trapped set}
K=\{G=\partial_rG_r=k_r=0, k\neq 0\}\, .
\end{equation}
Here,
\begin{align}
\label{eq:null cond}
G(x,k)=\rho^{2}g^{-1}_x(k,k)\, ,
\end{align}
so that $G(x,k)=0$ ensures that $(x,k)$ is lightlike, and 
\begin{align}
\label{eq:G_r}
G_r=\Delta_r k_r^2-\frac{\chi^2}{\Delta_r}\left((r^2+a^2)k_t+ak_\varphi\right)^2\, .
\end{align}
$\lambda_0$ can be chosen as 
\begin{align}
\label{eq:lambda_0}
\lambda_0=\min\{\lambda: \Delta_r(r_m)=a^2, \partial_r\Delta_r(r_m)=0, r_+<r_m<r_c\}\, .
\end{align}
 This choice guarantees that no $(x,k)$ with $k(\partial_t)=0$ are contained in $K$, see e.g. \cite{Dyatlov, Klein}.

Note that while the Boyer-Lindquist coordinates do not cover the rotation axis $\{\sin\theta=0\}$, the above expressions smoothly extend to the axis under an appropriate coordinate change \cite{Dyatlov}.

In the same way, $\Gamma^\pm$ can also be viewed as subsets of $T^*\MI$. For $(x,k)\in T^*\MI$, let $(\hat x,\hat k)$ denote the projection to $T^*(\rr_t\times\ss^2_{\theta,\varphi})$. If $0\leq\lambda<\lambda_0$, $\Gamma^\pm$ can be specified by \cite{Dyatlov}
\begin{align}
\label{eq:Gamma pm}
\Gamma^\pm=&\left\{(r,\hat x, k_r,\hat k)\in T^*\MI\vert \exists \tilde r(\hat x,\hat k):(\tilde r,\hat x,0,\hat k)\in K;\right.\\\nonumber
&\left. k_r=\pm \sgn(r-\tilde r)\sqrt{\tfrac{-G(r, \hat x, 0,\hat k)}{\Delta_r}}
\right\}
\end{align}
In the above, if $\tilde r$ exists, it is the unique solution to $-G(\tilde r,\hat x, 0,\hat k)=0$ for the given $\hat x$ and $\hat k$.

\section{The Unruh state on Kerr-de Sitter}
\label{sec:Unruh}
In this section, we will give a brief summary of the construction of the Unruh state in \cite{Klein} and show how the condition of small angular momentum $a$ can be removed by a more detailed analysis of the trapped set $K$.

\subsection{The free scalar field theory}
The Unruh state in \cite{Klein} has been constructed for the quantum field theory of a free, massive scalar field, governed by the Klein-Gordon equation
\begin{align*}
\left[\Box_g-m^2\right]\phi=0\, .
\end{align*} 
Here $\Box_g$ is the d'Alembertian for the metric $g$ introduced above, and $m>0$ is a constant, which can be identified as the mass of the scalar field.

Since the spacetimes $\tilde\cM$ and $\cM$ are globally hyperbolic \cite[Prop. 2.3]{Klein}, there exist unique advanced and retarded Green operators $E^\pm:C_0^\infty(\cM)\to C^\infty(\cM)$ and $\tilde{E}^\pm: C^\infty_0(\tilde \cM)\to C^\infty(\tilde \cM)$ for the Klein-Gordon operator on $\cM$ and $\tilde \cM$. Moreover, since $\cM$ can be isometrically embedded into $\tilde \cM$ while preserving the orientation, time orientation, and causal structure, the operators $\tilde{E}^\pm$ restricted to $C^\infty_0(\cM)$ agree with $E^\pm$. The difference between the retarded and advanced Green operator, $E=E^+-E^-$, is called the Pauli-Jordan propagator. 

There are various ways in which the algebra of observables for this theory can be formulated. Here, we choose the CCR-algebra $\cA$, the quotient of the free unital $*$-algebra generated by the smeared field operators $\Phi(f)$, $f\in C^\infty_0(\cM)$, with respect to the ideal defined by the relations
\begin{align*}
\Phi(\alpha f+h)&=\alpha \Phi(f)+\Phi(h)\\
 \Phi([\Box_g-m^2]f)&=0\\
 \left[\Phi(f),\Phi(h)\right]&=i\int\limits_\cM f E(h)\dVol_g 1\!\!1 \\ 
 \left(\Phi(f)\right)^*&=\Phi(\bar{f})\, ,
\end{align*}
for all $f,h\in C_0^\infty(\cM)$ and $\alpha\in\cc$.

A state on $\cA$ is a positive, normalised, linear map $\omega:\cA\to\cc$. It is quasi-free if it is completely determined by its two-point function $\omega(\Phi(f)\Phi(h))$. Consequently, a linear map $w:C_0^\infty(\cM)\times C_0^\infty(\cM)\to \cc$ defines a quasi-free state on $\cA$ if it satisfies
\begin{subequations}
\label{eq:conds on 2pt fct}
\begin{align}
w([\Box_g-m^2]f,h)&=w(f,[\Box_g-m^2]h)=0 \\
 w(\bar f, f)&\geq 0 \\
w(f,h)-w(h,f)&=i\int\limits_\cM f E(h)\dVol_g\, .
\end{align}
\end{subequations}
As usual, we will also require the two-point function $w$ to be a bi-distribution.

In addition, for a state to be physical, we demand that it satisfy the Hadamard property. For a quasi-free state on $\cA$, this can be formulated in terms of the wavefront set of the two-point function as the microlocal spectrum condition,
\begin{align*}
\WF'(w)&=\cC^+\\
\cC^+&=\{(x,k;y,l)\in T^*(\cM\times\cM):(x,k)\sim(y,l); k\triangleright 0\}\, .
\end{align*}
Here, $(x,k)\sim(y,l)$ means that $x$ and $y$ are connected by a null geodesic to which $k$ is cotangent at $x$ and $l$ cotangent at $y$. $k\triangleright 0$ indicates that $k$ is future directed. The primed wavefront set is the set
\begin{align*}
\WF'(w)=\{(x,k;y,l):(x,k;y,-l)\in \WF(w)\}\, .
\end{align*}
This property implies that the two-point function of any Hadamard state is of the form $w=w_0+H$, where $w_0$ is a smooth function on $\cM\times\cM$ and $H$, the Hadamard parametrix, is determined only by the theory and the local geometry.

\subsection{The Unruh state}
The Unruh state was originally constructed for the wave equation on Schwarzschild spacetime in the context of studying black-hole evaporation \cite{U}. It was later shown to be a well-defined Hadamard state by \cite{DMP}.
Extending the underlying idea to rotating black holes, the Unruh state for the Klein-Gordon equation on Kerr-de Sitter was defined in \cite{Klein} as follows.

\begin{definition}
\label{def:Unruh}
Let $f,h\in C_0^\infty(\cM)$. The Unruh state $\omega_U$ for the Klein-Gordon equation on Kerr-de Sitter is defined as the state specified by the two-point function
\begin{align}
\label{eq:def w}
w(f,h)&:=w_+(f,h)+w_c(f,h)\\\nonumber
&:=A_+(\tr_+ \tilde{E}(f), \tr_+ \tilde{E}(h))+A_c(\tr_c \tilde{E}(f), tr_c \tilde{E}(h))\, ,
\end{align}
where $\tr_i:C^\infty(\tilde \cM)\to \cD'(\sH_i)$, $\phi\mapsto \phi\vert_{\sH_i}$ is the trace map to the horizon $\sH_i$ and the bi-distributions $A_i:C_0^\infty(\sH_i)\times C_0^\infty(\sH_i)\to \cc$, for $i\in \{+,c\}$, are given by
\begin{align}
\label{eq:A_i}
A_i(F,H)=-\lim\limits_{\epsilon\to 0} \frac{r_i^2+a^2}{\chi\pi}\int\limits_{\rr^2\times \ss^2_{\theta,\varphi_i}}\frac{F(L_i,\theta,\varphi_i)H(L'_i,\theta,\varphi_i)}{(L_i-L'_i-i\epsilon)^2}dL_idL'_id\Omega_i^2\,.
\end{align}
Here, $L_+=U_+$, $L_c=V_c$, and $d\Omega_i^2$ is the standard volume element of the unit two-sphere $\ss^2_{\theta,\varphi_i}$. 
\end{definition}

The first important result of \cite{Klein} was to show that
\begin{proposition}
\label{prop:well-def}
The bi-linear map $w$ in Definition \ref{def:Unruh} is well-defined and constitutes the two-point function of a quasi-free state on $\cA$.
\end{proposition} 

\begin{remark}
A key to showing Proposition \ref{prop:well-def} are decay results for solutions to the Klein-Gordon equation on Kerr-de Sitter provided by \cite{HV}. As also analysed in \cite{HWZ}, the results of \cite{HV} entail that close to the horizons $\sH_i$, $i\in\{+,c\}$,
\begin{align}
\label{eq:decay horizons}
\abs{\partial_{L_i}^n \tilde{E}(f)}\lesssim C \abs{L_i}^{-(n+\alpha/\kappa_i)}\norm{f}_{C^m}\, ,
\end{align}
where $m$ depends on $n$ and the constant $C$ depends on $n$ and $\supp f$. The $C^m$-norm is defined as
\begin{align*}
\norm{f}_{C^m}=\max\limits_{\abs{\beta}\leq m}\sup\limits_{x\in K}\abs{\left(\prod\limits_jV_j^{\beta_j}\right)f(x)}\, ,
\end{align*}
where $K\subset \cM$ is a compact region containing $\supp f$, $V_j$, $j\in\{1,\dots,4\}$ are linearly independent, smooth vector fields on $K$ and $\beta$ is a multi-index. The constant $\alpha$ is the spectral gap of the massive Klein-Gordon equation. 

In the region $\{r_++\delta<r<r_c-\delta\}\subset\MI$, the results of \cite{HV} provide the estimate
\begin{align}
\label{eq:decay MI}
\abs{\partial^n \tilde{E}(f)}\lesssim C e^{\alpha t}\norm{f}_{C^m}\, , \quad \partial\in \{\partial_t,\partial_r,\partial_\theta,\partial_\varphi\}\, .
\end{align}

These decay results are sufficient for the purpose of showing the well-definedness of the Unruh state if the spectral gap is (strictly) positive, i.e., $\alpha>0$. While this is expected to hold for all subextreme Kerr-de Sitter spacetimes, it has only been shown for small $\abs{a}$ \cite{Dyatlov:QNM} or small $\lambda$ \cite{Hintz:QNM}. In the rest of this work, we will always work under the assumption of mode stability and a positive spectral gap.
\end{remark}

\begin{proof}
In the following, we will roughly sketch the proof of Proposition \ref{prop:well-def}.

Below, let $f, h\in C_0^\infty(\cM)$. In a first step, we show that $A_i$ acting on $\tr_i \tilde{E}(f)\otimes \tr_i\tilde{E}(h)$ leads to a well-defined, finite expression, despite $\tr_i \tilde{E}(f)$ not being compactly supported on $\sH_i$. 

To do so, we fix some $\epsilon>0$ and consider $A_{i,\epsilon}$, which is given by \eqref{eq:A_i} without the limit of $\epsilon\to 0$. We then aim to estimate the absolute value of 
\begin{align*}
A_{i,\epsilon}(\tr_i\tilde{E}(f),\tr_i\tilde{E}(h))\, .
\end{align*}
First, we use the decay estimates \eqref{eq:decay horizons} to justify a partial integration in $L_i$ and $L'_i$. Then, we split the integration region into different domains on which \eqref{eq:decay horizons} and known results on the sequence $\log (y-i\epsilon)$ of $L^1_{loc}$-functions can be applied to find an upper bound of the form 
\begin{align*}
\abs{A_{i,\epsilon}(\tr_i\tilde{E}(f),\tr_i\tilde{E}(h))}\leq C \norm{f}_{C^m}\norm{h}_{C^m}\, 
\end{align*}
 for some $m\in \nn$. Importantly, the bound is uniform in $\epsilon$, and the constant $C$ depends only on $\supp f \cup \supp h$. This shows that 
the combination $A_i(\tr_i \tilde{E}(\cdot), \tr_i\tilde{E}(\cdot))$ is a continuous map from $C_0^\infty(\cM)\times C_0^\infty(\cM)$ to $\cc$, i.e., a bi-distribution.

Next, since the kernel of $\tilde{E}$ restricted to $C_0^\infty(\cM)$ is $[\Box_g-m^2]C_0^\infty(\cM)$, $w(f,h)$ defined in \eqref{eq:def w} satisfies the first condition in \eqref{eq:conds on 2pt fct}.

To show the second property in \eqref{eq:conds on 2pt fct}, the positivity, we rewrite $w_i(f,h)$ using Fourier transform in $L_i$. More specifically, let 
\begin{align*}
\mu_i=\frac{2\eta (r_i^2+a^2)}{\chi}d \eta d\Omega_i^2\, ,
\end{align*}
be a measure on $\rr_+\times\ss^2$. Given any function $\psi\in C^\infty(\rr;\rr)$ satisfying $\psi(x)=1$ for $x>x_+$, $\psi(x)=0$ for $x<x_-$, and $x_-<x_+<0$, define the map $K_i: C_0^\infty(\cM)\to L^2(\rr_+\times \ss^2, \mu_i)$  as
\begin{align*}
K_i(f):=\mathbf{1}_{\rr_+}(\eta)\cF_i(\psi \tr_i(\tilde{E}(f)))+\mathbf{1}_{\rr_+}(\eta)\cF_i((1-\psi) \tr_i(\tilde{E}(f)))\, , \\
\cF_i(F)(\eta,\theta,\varphi):=(2\pi)^{-\tfrac{1}{2}}\int\limits_\rr e^{i\eta L_i}F(L_i,\theta,\varphi_i)dL_i\, .
\end{align*}
Then one has
\begin{align}
\label{eq:Fourier-version w_i}
w_i(f,h)=\IP{K_i(\bar{f}), K_i(h)}_{L^2(\rr_+\times\ss^2,\, \mu_i)}\, .
\end{align}
To show that the right-hand side of \eqref{eq:Fourier-version w_i} is well-defined and that \eqref{eq:Fourier-version w_i} holds, one first checks by explicit computation that the map $\mathbf{1}_{\rr_+}(\eta)\circ \cF $ is an isometric embedding of $(C_0^\infty(\sH_i), A_i(\bar{\cdot},\cdot))$ into $L^2(\rr_+\times\ss^2; \mu_i)$, which extends by linearity and continuity to all of $\overline{(C_0^\infty(\sH_i), A_i(\bar{\cdot},\cdot))}$. With some additional work, one can show that the extended map is indeed an isomorphism \cite{DMP}. 

Next, one needs to show that $\tr_i\tilde{E}(f)$ can be identified with an element of $\overline{C_0^\infty(\sH_i)}^{A_i}$ for all $f\in C_0^\infty(\cM)$. Following \cite{DMP}, this can be done by separating $\tr_i\tilde{E}(f)$ into a compactly supported, smooth piece and a piece supported on $\sH_i^-:=\sH_i\cap\{L_i<0\}$, with support disjoint from $\{L_i=0\}$.  For the piece supported on $\sH_i^-$, the result follows from a combination of the decay results \eqref{eq:decay horizons} with a coordinate transform to $l_i=-\kappa_i^{-1}\log\abs{L_i}$.

Finally, it remains to show the third property in \eqref{eq:conds on 2pt fct}. By a standard result, see e.g. \cite{Dimock}, one can write 
\begin{align*}
\int\limits_\cM f E(h) \dVol_g=\int\limits_\Sigma \left[E(f)\nabla_aE(h)-E(h)\nabla_a E(f)\right]n_\Sigma^a\dVol_\gamma\, ,
\end{align*}
where $\Sigma$ is a Cauchy surface of $\cM$ with future-directed unit normal $n_\Sigma$ and induced metric $\gamma$. This is independent of the choice of Cauchy surface, since the current $j_a=E(f)\nabla_aE(h)-E(h)\nabla_a E(f)$ is conserved, $\nabla_aj^a=0$. Again, making use of the decay results \eqref{eq:decay horizons}, and in addition  \eqref{eq:decay MI}, one can evaluate this on a sequence of Cauchy surfaces approaching $\sH_+\cup i^-\cup \sH_c$, and show that there is no contribution coming from $i^-$. Comparing the current evaluated on $\sH_+\cup\sH_c$ to the antisymmetric part of $w$ then gives the desired result upon partial integration.
\qeds
\end{proof}

\begin{remark}
A byproduct of the proof of the second property is that for $f,h\in C_0^\infty(\MI)$, the two-point function can be written as
\begin{align}
\label{eq:Fourier-version w_i I}
w_i(f,h)=\IP{K_{i,I}(\bar{f}), K_{i,I}(h)}_{L^2(\rr\times\ss^2, \mu_{i,I})}\, ,\\
\mu_{i,I}=\frac{(r_i^2+a^2)}{\chi}\frac{\xi e^{\pi\xi/\kappa_i}}{\sinh (\pi\xi/\kappa_i)} d \xi d\Omega_i^2\, ,
\end{align}
and $K_{i,I}: C_0^\infty(\MI)\to L^2(\rr\times \ss^2, \mu_{i,I})$ given by
\begin{align*}
K_{i,I}(f):=(2\pi)^{-\tfrac{1}{2}}\int\limits_\rr e^{i\xi l_i}(\tr_i\tilde{E})(f)(l_i,\theta,\varphi_i)dl_i\, .
\end{align*}
\end{remark}

While the proof of Proposition \ref{prop:well-def} could be extended from Schwarzschild \cite{DMP} to Kerr-de Sitter in a relatively straightforward manner with the help of the decay results of \cite{HV}, the proof of the Hadamard property in \cite{DMP} relies on the presence of a global timelike Killing field in the black hole exterior, which is absent in the case of Kerr-de Sitter.
Nonetheless, it was shown in \cite{Klein}

\begin{theorem}
\label{thm:Hadamard}
There exists an $a_0>0$ and a $\lambda_0>0$ so that under the assumption of mode stability, the Unruh state for the Klein-Gordon field on Kerr-de Sitter is Hadamard for all $0<\abs{a}<a_0$ and all $0<\lambda_0$.
\end{theorem}

\begin{proof}
Again, we give a rough sketch of the proof. First, due to the first property in \eqref{eq:conds on 2pt fct}, the Propagation of Singularities Theorem (in the version of \cite[Lemma 6.5.5]{DH}) tells us that $WF'(w)\subset \cN\times\cN$ and allows us to propagate regularity (and singularities) along the bicharacteristics of $[\Box_g-m^2]$, in other words along the lift of null geodesics to $T^*\cM$. Moreover, let us note that due to the positivity of the two-point function and an application of the Cauchy-Schwarz theorem, together with Propagation of Singularities, it suffices to show that 
\begin{align*}
\WF'(w)\cap \Delta_{T^*\cM\times T^*\cM}\subset \cN^+\times\cN^+\, .
\end{align*}
Here $\Delta_{T^*\cM\times T^*\cM}$ is the diagonal in $T^*\cM\times T^*\cM$.

The proof of the Hadamard property can now be split into two cases, null geodesics ending at $\sH_+\cup\sH_c$ and forwards or backwards trapped geodesics in $\Gamma$ that do not approach $\sH_+\cup\sH_c$.

The first case can be dealt with by using the explicit form \eqref{eq:def w} of the two-point function. Focussing on a particular bicharacteristic, one can use cutoff functions to single out a relevant piece and compute its wavefront set directly. For this case, it then remains to show that the bicharacteristic under consideration does not appear in the wavefront set of the remainder terms.
In \cite{Klein}, this was handled by working on $\tilde{\cM}$ and making use of various properties of the Pauli-Jordan propagator $\tilde{E}$ viewed as a map from $C_0^\infty(\tilde \cM)/[\Box_g-m^2]C_0^\infty(\tilde \cM)$ to solutions to the Klein-Gordon equation. This was combined with a further decomposition using cutoff functions, an analysis of the support properties of the various pieces, and the decay results from \cite{HV} to show the desired result.

To handle the second case for the Unruh state on Schwarzschild, the authors of \cite{DMP} observed that the two pieces of their two-point function, when restricted to $C_0^\infty(\MI)\times C_0^\infty(\MI)$, satisfy the analyticity properties of a KMS/ ground state with respect to a globally timelike Killing field in $\MI$. This observation allowed the authors to conclude the Hadamard property in the same way as for passive states \cite{SV}.

In the case of Kerr-de Sitter, there is no globally timelike Killing vector field on $\MI$, and this is not possible. However, using the formulation \eqref{eq:Fourier-version w_i I} of $w$, one can show that each of the two pieces $w_i$ satisfies a KMS analyticity property with temperature $2\pi \kappa_i^{-1}$, albeit with respect to two different Killing vector fields. Concretely, the piece $w_i$ obeys a KMS-type property with respect to $v_i$. This follows from direct inspection of \eqref{eq:Fourier-version w_i I}, again making use of the decay estimates \eqref{eq:decay horizons}.

The KMS-like property allows us to follow the proof of \cite[Theorem 5.1]{SV} for $w_+$ and $w_c$ separately. In this way, one can show that points of the form $(x,k;x,k)$ with $x\in \MI$ and $k$ such that $k(v_i)<0$ cannot be contained in $\WF'(w_i)$.

By analysing the null geodesics on Kerr-de Sitter, one realises that there are some $0<a_0$ and $0<\lambda_0$, so that for any $0<\abs{a}<a_0$ and $0<\lambda<\lambda_0$, any bicharacteristic in $\Gamma$ intersects a region where both $v_+$ and $v_c$ are future-directed timelike. As a result, $k(v_i)<0$ entails that $(x,k)\in \cN^-$. Taken together with the fact that $\WF'(w_i)\subset \cN$, the above result then entails for the case of trapped bicharacteristics
\begin{align*}
\WF'(w_i)\cap \Delta_{\Gamma\times \Gamma}\subset \cN^+\times\cN^+\, ,
\end{align*}
as desired. Together with the first case, this concludes the proof of the Hadamard property.
\qeds
\end{proof}

\subsection{The Hadamard property for large $a$}
A crucial takeaway of the above is that apart from the issue of mode stability, the assumption of small angular momentum only enters in the very last part of the proof. 

The last part of the proof combined the ideas from \cite{DMP} with those of \cite{GHW}. In the latter work, the authors constructed the Unruh state for free, massless fermions and proved its Hadamard property on subextreme Kerr spacetimes with sufficiently small $a$. In fact, in this work, the smallness condition for the angular momentum enters in the same way as above: The two-point function is split into two separate pieces corresponding to $w_c$ and $w_+$. It is shown that points of the form $(x,k;x,-k)$ with $x\in \MI$ and  $k(\partial_t)<0$ cannot be in the wavefront set of the first piece, while those with $k(v_+)<0$ cannot be in the wavefront set of the second piece. An analysis of the trapped set of Kerr spacetime equivalent to the analysis above then allows the authors to conclude the Hadamard property, given that $\abs{a}$ is sufficiently small.

However, a more detailed analysis of the trapped set in this case reveals that the smallness condition on $\abs{a}$ can be circumvented.
This was elaborated for the fermionic Unruh state on Kerr in \cite{KHa}. Since the argument presented there is purely geometrical, it should also hold for other theories on the same spacetime. In fact, the key geometric result of \cite{KHa} can also be generalised to Kerr-de Sitter.

\begin{proposition}
\label{prop:fb trapped set}
Let $(\cM,g)$ be the physical part of a subextreme Kerr- or Kerr-de Sitter spacetime with cosmological constant $\lambda$ and angular momentum per unit mass $a$. Then there exists a $0<\lambda_0$, so that for all $0\leq\lambda<\lambda_0$, if the geodesic corresponding to $(x,k)\in T^*\cM$ lies in $\Gamma$ and satisfies $k(v_+)\geq 0$ or $k(v_c)\geq 0$, then $(x,k)\in \cN^+$.
\end{proposition}

Before proving this proposition, we first show some intermediate results. 
In Lemma \ref{lemma:future past directed}, we characterise future- and past-directed null covectors. Since the key to analysing the set $\Gamma$ is the trapped set $K$, we use this characterisation and the results in \cite{Dyatlov} to show that $(x, k)\in K$ is future-directed if the $t$-component of $k$ is positive. In Lemma~\ref{lemma:Npm K overlap}, we show that for $(x, k)\in K$ the property $k(v_+)\geq 0$ or $k(v_c)\geq 0$ implies $k_t>0$ and therefore that $k$ is future directed. To prove Proposition~\ref{prop:fb trapped set}, we utilise these results and the connection between the trapped set $K$ and $\Gamma$. 

As noted in Section~\ref{sec:null geos}, $\Gamma$ can be viewed as a subset of $T^*\MI$. We will thus assume that $(x,k)\in T^*\MI$ and work in Boyer-Lindquist coordinates.
Note that a null covector $(x,k)\in T^*\MI$ can be determined by $x$ and $(k_r,k_\theta,k_\varphi)\in \rr^3$ by setting
\begin{align}
\label{eq:xit null cond}
k_t=-\frac{\mu}{\sigma^2}a k_\varphi\pm\sqrt{\frac{\Delta_r\Delta_\theta}{\sigma^2}\left[\frac{\rho^4}{\sigma^2\sin^2\theta}k_\varphi^2+\frac{\Delta_\theta}{\chi^2}k_\theta^2+\frac{\Delta_r}{\chi^2}k_r^2\right]}\, .
\end{align}

From this, one obtains that 
\begin{lemma}
\label{lemma:future past directed}
$(x,k)$ is future directed (past directed) if one chooses the upper (lower) sign in \eqref{eq:xit null cond}.
\end{lemma}
\begin{proof}
By our definition, $(x,k)$ is future (past) directed if and only if $k(-\nabla t)$ is bigger (smaller) than zero, see also \cite[Lemma 6.3]{KHa} and \cite{Dyatlov}. It then follows immediately that for $(x,k)$ with $x\in\MI$ and $k=(k_t,k_r,k_\theta,k_\varphi)$ satisfying \eqref{eq:xit null cond}, one has
\begin{align*}
k(-\nabla t)=&\frac{\chi^2\sigma^2}{\rho^2\Delta_\theta\Delta_r}\left(k_t+\frac{\mu}{\sigma^2}ak_\varphi\right)\\
=&\pm \frac{\chi^2\sigma^2}{\rho^2\Delta_\theta\Delta_r}\sqrt{\frac{\Delta_r\Delta_\theta}{\sigma^2}\left[\frac{\rho^4}{\sigma^2\sin^2\theta}k_\varphi^2+\frac{\Delta_\theta}{\chi^2}k_\theta^2+\frac{\Delta_r}{\chi^2}k_r^2\right]}\, ,
\end{align*}
so $\pm k(-\nabla t)>0$ if and only if one chooses the plus (minus) sign in \eqref{eq:xit null cond}.
\qeds
\end{proof}

Next, we recall that there exists a $0<\lambda_0$, so that for $0\leq \lambda<\lambda_0$, the trapped set in Kerr(-de Sitter) can be characterised by \cite{Dyatlov}
\begin{align*}
K=\{(x,k)\in T^*\MI:G=\partial_r G_r=k_r=0,\, \, k\neq 0\}\subset T^*\MI\, ,
\end{align*}
where $G$ and $G_r$ are given in \eqref{eq:null cond} and \eqref{eq:G_r} respectively.

Note that $\partial_rG_r=0$, $ k\neq 0$ is equivalent to (see \cite{Dyatlov})
\begin{align}
\label{eq:drGr=0}
4r\Delta_r k_t-\left((r^2+a^2) k_t+a^2 k_\varphi\right)\Delta_r'=0\, .
\end{align}
From now on, we restrict to $0\leq \lambda <\lambda_0$ as described in \eqref{eq:lambda_0}. Recall that this condition ensures that no $(x,k)$ with $k_t=0$ are contained in $K$. 
 In this case, the trapped set $K$ has the following properties \cite{Dyatlov} 
 \begin{proposition}
\label{prop:K props}
The trapped set on Kerr(-de Sitter) with sufficiently small $\lambda$ exhibits the following properties:
\begin{enumerate}
\item On $K$, one has $\partial_r\Delta_r>0$ . This also implies that on $K$, \begin{align*}
     k_t\left((r^2+a^2) k_t+a k_\varphi\right)>0\, .
\end{align*}
\item $K\cap\{ \pm k_t>0\}\subset \cN^\pm$.
\end{enumerate}
\end{proposition}
\begin{proof}
In the Kerr case, one has $\partial_r\Delta_r>0$ on all of $\MI$. Therefore, the first point follows for $\lambda=0$ after multiplying \eqref{eq:drGr=0} by $k_t$, taking into account that $k_t\neq 0$. For $\lambda\neq 0$, the equation follows by a perturbation argument combined with the fact that $K$ is contained in a compact subset of $\cM$ that continuously depends on $\lambda$ \cite{Dyatlov}, as well as the fact that $k_t\neq 0$ on $K$ for all $\lambda<\lambda_0$. To show the second point, let $k_t>0$ and assume that $(x,k)\in K$ is past-directed. Then by Lemma \ref{lemma:future past directed}, $k_t$ must be given by the lower sign in \eqref{eq:xit null cond}. Since $k_t>0$, this requires $a k_\varphi<0$ and leads to the upper bound 
\begin{align}
\label{eq:est}
k_t\leq \abs{\frac{\mu}{\sigma^2}a k_\varphi}\leq \abs{\frac{ak_\varphi}{r^2+a^2}}\leq \frac{\abs{ak_\varphi}}{r_+^2+a^2}\, ,
\end{align}
which follows from  
\begin{align}
\frac{1}{r^2+a^2}-\frac{\mu}{\sigma^2}=\frac{\Delta_r\rho^2}{\sigma^2(r^2+a^2)}\, .
\end{align}
Hence
\begin{align*}
k_t((r^2+a^2)k_t+ak_\varphi)<0\, ,
\end{align*}
in contradiction to the first point. The argument works in the same way for $k_t<0$ and $(x,k)$ future-directed.
\qeds
\end{proof}
With this, one can prove
\begin{lemma}
\label{lemma:Npm K overlap}
Let $ i\in\{+,c\}$, and $0\leq\lambda<\lambda_0$. Then
 $K\cap\{k_t<0\}\cap\{k(v_i)\geq 0\}=\emptyset$ .
 \end{lemma}
 \begin{proof}
 Let $(x,k)\in K\cap\{k_t<0\}\cap\{k(v_i)\geq 0\}$.
 To begin, we note that $k(v_i)\geq 0$ corresponds to $k_t\geq -\tfrac{ak_\varphi}{r_i^2+a^2}$. Since $k_t<0$, one must have $ak_\varphi>0$. 
 We also note that for any fixed $ak_\varphi>0$, the set $\{-\tfrac{ak_\varphi}{r_c^2+a^2}\leq k_t<0\}$ is contained in the set $\{-\tfrac{ak_\varphi}{r_+^2+a^2}\leq k_t<0\}$. We can thus focus on showing the result for the larger of the two sets. 
 
 Next, as described in \cite{Dyatlov}, the condition $\partial_rG_r=0$ is reduced to
 \begin{align}
 -\left(4r\Delta_r-(r^2+a^2)\partial_r\Delta_r\right)&=\partial_r\Delta_r\frac{-ak_\varphi}{k_t}\, .
\end{align}

Combining this with $k(v_+)\geq 0$, the first point of Proposition \ref{prop:K props}, in particular $\partial_r\Delta_r>0$ on $K$, and making use of the fact that $k_t$ and $ak_\varphi$ are of opposite sign, so that the right-hand side above is positive, one obtains the estimate 
\begin{align}
-\left(4r\Delta_r-(r^2+a^2)\partial_r\Delta_r\right)\geq\partial_r\Delta_r(r_+^2+a^2)
\end{align}
which may be simplified to 
\begin{align}
 P(r):= 4r\Delta_r-(r^2-r_+^2)\partial_r\Delta_r\leq 0\, .
\label{eq:P neg}
\end{align}
In the case $\lambda=0$, it was shown in \cite[Lemma 6.5]{KHa} that $P(r)$ is monotonously increasing for $r>r_+$ and vanishes at $r=r_+$. 
To analyse $P(r)$ on $(r_+,r_c)$ in the case $\lambda\neq 0$, we note that $\Delta_r$ and $\partial_r\Delta_r$ can be written as
\begin{align*}
\Delta_r&=-\lambda \prod\limits_{i=0}^{3}(r-r_i)\, ,\\
\partial_r\Delta_r&=-\lambda \sum\limits_{i=0}^{3}\prod\limits_{j\neq i}(r-r_j)\, ,
\end{align*}
where we have identified $r_1=r_-$, $r_2=r_+$, $r_3=r_c$, and $r_0=-(r_1+r_2+r_3)$.
Thus, we have, taking into account that $r_0<r_1<r_2<r<r_3$ on $\MI$, 
\begin{align*}
P(r)&=4r\Delta_r+\lambda(r+r_2)(r-r_2)\sum\limits_{i=0}^{3}\prod\limits_{j\neq i}(r-r_j)\\
&=4r\Delta_r+\lambda(r+r_2)(r-r_2)^2(r-r_1)(2r-r_3-r_0)-(r+r_2)\Delta_r\\
&\hphantom{=}+\lambda(r+r_2)(r-r_2)^2(r-r_3)(r-r_0)\\
&>(3r-r_2)\Delta_r+\lambda(r+r_2)(r-r_2)^2(r-r_1)(2r+r_1+r_2)-(r+r_2)\Delta_r\\
&=2(r-r_2)\Delta_r+\lambda(r+r_2)(r-r_2)^2(r-r_1)(2r+r_1+r_2)>0\, ,
\end{align*}
on $\MI$. In the third step, we have used $(r-r_1)>(r-r_2)$.
Therefore, $P(r)>0$ on $(r_+,r_c)$, in contradiction to \eqref{eq:P neg}.
\qeds
\end{proof}

We can now show Proposition \ref{prop:fb trapped set}.\\
\begin{proof}{\it{Proof of Proposition \ref{prop:fb trapped set}:}}
Let $\lambda_0$ be given by \eqref{eq:lambda_0}, and assume $0\leq \lambda<\lambda_0$. Let us assume $(x,k)=(r,\hat x, k_r, \hat k)\in \cN^-\cap\Gamma\cap\{k(v_i)\geq 0\}$. Here, $(\hat x, \hat k)$ is the projection of $(x,k)$ to $T^*(\rr_t\times\ss^2_{\theta,\varphi})$.

Then, recalling the form of $\Gamma^\pm$ in \eqref{eq:Gamma pm}, there must be a $\tilde r$ so that $(\tilde r, \hat x, 0,\hat k)\in K$. By Lemma~\ref{lemma:Npm K overlap}, $k(v_i)\geq 0$ then implies $k_t>0$. However, following the proof of the second point of Proposition~\ref{prop:K props}, if $k_t>0$ and $k$ is past directed, so that $k_t$ satisfies the bound \eqref{eq:est}, then  $(\tilde r, \hat x, 0,\hat k)$ cannot be in $K$ for any $\tilde{r}\in (r_+,r_c)$ (or $\tilde{r}\in (r_+,\infty)$ for $\lambda=0$), in contradiction to the assumption that $(x,k)\in \Gamma^\pm$.
\qeds
\end{proof}

From this result, it now follows immediately that
\begin{theorem}
Let $0<\lambda<\lambda_0$, with $\lambda_0$ given by \eqref{eq:lambda_0}. Let $(\cM, g)$ be a subextreme Kerr-de Sitter spacetime with cosmological constant $\lambda$, and let $\omega_U$ be the Unruh state for the Klein-Gordon field on $\cM$ as defined in Definition~\ref{def:Unruh}. If mode stability holds, then the Unruh state is Hadamard on $\cM$. In particular, this shows that the Unruh state is Hadamard up to, but not including, the inner horizon of the black hole.
\end{theorem}

\begin{proof}
Following the proof of Theorem~\ref{thm:Hadamard}, we see that the only case that we need to consider is geodesics in $\Gamma$ that do not approach $\sH_+$ or $\sH_c$. If such a geodesic is specified by some $(x,k)\in T^*\MI$, then by the proof of Theorem~\ref{thm:Hadamard}, $(x,k;x,k)$ cannot be in $\WF'(\omega)$ if $k(v_+)<0$ and $k(v_c)<0$. In other words, for $(x,k;x,k)\in\WF'(\omega)\cap(\Gamma\times\Gamma)$, either $k(v_+)\geq 0$ or $k(v_c)\geq 0$. By Proposition~\ref{prop:fb trapped set}, this then implies that $(x,k)\in\cN^+$, showing the desired result.
\qeds
\end{proof}

\section{Application: The inner horizon}
\label{sec:appli}
The state constructed above now allows us to study different effects of the quantum Klein-Gordon field on Kerr-de Sitter. One example is the behaviour of the stress-energy tensor at the inner horizon $\sH_-$ of $\cM$, see Section~\ref{sec:Kruskal}. This quantity plays an important role in determining the influence of quantum effects on the (non-)validity of the strong cosmic censorship conjecture, or, in other words, the change of geometry near the inner horizon under perturbations by a quantum field.

\subsection{The stress-energy tensor in the Unruh state}
For theories such as the free scalar field theory, the classical stress-energy tensor typically consists of terms of the form $D_1\phi(x)D_2\phi(x)$, where $D_1$ and $D_2$ are differential operators. Therefore, the corresponding quantum observable needs to be renormalised. In curved spacetime, one can make use of the Hadamard point-split renormalisation procedure. In this procedure, one first computes $D_1\phi(x)D_2\phi(y)$, with $x$ and $y$ spacelike separated. Before taking the coinciding point limit, one subtracts the Hadamard parametrix $H$, the non-smooth part of the two-point function, acted on by $D_1$ and $D_2$. In this way, one obtains a locally and covariantly renormalised stress-energy tensor. 
However, in practice, this procedure is very hard to realise. The reason is that the forms of the two-point function and the Hadamard parametrix used for numerical calculations are not compatible with each other.

There are different methods to circumvent this problem. One is a pragmatic approach to point-split renormalisation \cite{LO, LO2}.
Another one is particularly well-suited if one is only interested in the leading behaviour of an observable such as the stress-energy tensor at the inner horizon $\sH_-$. In this case, instead of computing the expectation value in one state, one computes the difference of expectation values between the desired state and an unphysical comparison state, which is Hadamard in a neighbourhood of the horizon $\sH_-$, in contrast to the Unruh state. One may choose the comparison state in such a way that taking the limit onto $\sH_-$ leads to a comparatively simple formula that is well-suited for numerical evaluation.

This program was carried out in \cite{KSCH} to study the leading divergences of different components of the stress-energy tensor of the massive scalar in the Unruh state on Kerr-de Sitter in the case $m^2=2\lambda$.
Similar to previous results in the Unruh state on various black hole spacetimes \cite{HKZ,ZLO,ZCOO}, \cite{KSCH} indeed discovered a divergence of the form $(r-r_-)^{-n}$ for the components considered. Here, $n$ is the number of $r$-derivatives appearing in the operators $D_1$ and $D_2$ for the corresponding component. In all these results, the sign of the prefactor of the highest-order divergence was found to depend on the spacetime parameters and, in the case of rotating black holes, on the latitude $\theta$ on the inner horizon.

In light of these results, it seems likely that quantum effects can have a significant influence on the geometry near the inner horizon of a black hole. However, settling this question will require at least a more self-consistent consideration of semi-classical gravity.

\subsection{Universality at the inner horizon}

If the nature of the quadratic divergence of the stress-energy tensor presented above depends strongly on the choice of state, then these results do not contribute much towards understanding the generic situation. However, grasping the generic behaviour of the quantum field is crucial to determine its influence on the strong cosmic censorship conjecture, and for the results to be stable enough to be meaningful.

Fortuitously, it can be shown that the divergence of the stress-energy tensor, and in fact of a larger class of observables,
has a certain universality property. This result was first obtained for a class of charged black hole spacetimes in \cite{HWZ}, and generalised to all charged asymptotic de-Sitter black holes and Kerr-de Sitter spacetimes by \cite{HK}.

Let us focus on Kerr-de Sitter and a point $x_0\in\sH_-$. Let $\cU$ be a precompact neighbourhood of $x_0$ in the Kruskal manifold $\Mm$, and assume $\cU$ to be covered by the coordinate system $((r-r_-), y)$, where $y$ parametrises $\sH_-\cap \cU$.

Then the classical observables considered in \cite{HK} are of the form
\begin{align}
\label{eq:op form}
\cO^{j_1,j_2}_{\gamma_1,\gamma_2}(x)=\partial_r^{j_1}\partial_y^{\gamma_1}\Phi(x)\partial_r^{j_2}\partial_y^{\gamma_2}\Phi(x)\, ,\quad x\in \cM\cap \cU\, ,
\end{align}
or combinations thereof with smooth coefficients that extend smoothly to $\overline{\cU}$. Here, $j_{1,2}\in \nn$, and $\gamma_{1,2}\in\nn^3$ are multi-indices. In \cite{HK}, they were restricted to $j_1+j_2+\abs{\gamma_1}+\abs{\gamma_2}\leq 2$, but we drop this restriction. This class includes the stress-energy tensor components in the coordinate system $((r-r_-),y)$.

The quantised versions of these observables are not contained in the algebra of observables defined earlier. However, after renormalisation, they can be viewed as operator-valued distributions with values in an enlarged operator algebra, to which any Hadamard state can be extended \cite{HW}. In fact, denoting the renormalised quantum observable as $\colon \cO^{j_1,j_2}_{\gamma_1,\gamma_2}(x)\colon$, for any Hadamard state on the algebra $\cA$ and any $x\in\cM\cap\cU$, $\omega(\colon\cO^{j_1,j_2}_{\gamma_1,\gamma_2}(x)\colon)$ will be well-defined and finite. For the case when $x\to x_0$, it follows as a corollary of the results in \cite{HK}

\begin{proposition}{\cite[Proposition 4.1]{HK}}
Let $(\cM,g)$ be a subextreme Kerr-de Sitter spacetime, and assume that mode stability holds and that there exists a positive spectral gap $\alpha>0$ for the (massive) wave equation. Let $\cU$ as above be covered by coordinates $(r-r_-, y)$.\\
Let $\omega_1$ and $\omega_2$ be Hadamard states for the free scalar field on $\cM$, and let $\cO^{j_1,j_2}_{\gamma_1,\gamma_2}(x)$ be an observable of the form described in \eqref{eq:op form}.
Then there are constant $C_{j_1,j_2,\gamma_1,\gamma_2}$, so that for some $0<\beta'<\min(\alpha/\kappa_-,1)$ the estimates
\begin{align}
\abs{(r-r_-)^{j_1+j_2-\beta'}\left(\omega_1(\colon \cO^{j_1,j_2}_{\gamma_1,\gamma_2}(x)\colon)-\omega_2(\colon\cO^{j_1,j_2}_{\gamma_1,\gamma_2}(x)\colon)\right)}\leq C_{j_1,j_2,\gamma_1,\gamma_2}\,
\end{align} 
for $j_1+j_2>0$, and 
\begin{align}
\abs{\omega_1(\colon \cO^{j_1,j_2}_{\gamma_1,\gamma_2}(x)\colon)-\omega_2(\colon\cO^{j_1,j_2}_{\gamma_1,\gamma_2}(x)\colon)}\leq C_{j_1,j_2,\gamma_1,\gamma_2}
\end{align} 
if $j_1+j_2=0$ hold uniformly in $x\in\cU\cap\cM$.
\end{proposition}

\begin{proof}
The starting point of the proof is that for two Hadamard states $\omega_1$ and $\omega_2$ on $\cA$, the bi-distribution 
\begin{align*}
W(f,h)=\omega_1(\Phi(f)\Phi(h))-\omega_2(\Phi(f)\phi(h))\, ,\quad f,h\in C_0^\infty(\cM)
\end{align*}
can, in fact, be written as
\begin{align*}
W(f,h)=\int\limits_{\cM\times\cM}W(x,y)f(x)h(y)\dVol_g(x)\dVol_g(y)
\end{align*}
with $\dVol_g$ the volume form induced by the metric and $W(x,y)$ a smooth function on $\cM\times\cM$. As a consequence, one can write for $x,x'\in \cM\cap\cU$
\begin{align*}
\omega_1(\colon \cO^{j_1,j_2}_{\gamma_1,\gamma_2}(x)\colon)-\omega_2(\colon\cO^{j_1,j_2}_{\gamma_1,\gamma_2}(x)\colon)=\lim\limits_{x'\to x}g(x,x')\partial_r^{j_1}\partial_y^{\gamma_1}\partial_{r'}^{j_2}\partial_{y'}^{\gamma_2}W(x,x')\, ,
\end{align*}
where $g(x,x')$ should be understood as the proper power and component of the parallel transport bi-tensor mapping $T_{x'}\cM$ to $T_x\cM$. The parallel transport bi-tensor ensures that the expression transforms like a tensor at $x$ under coordinate transforms, and simply reduces to the identity in the limit.

Next, making use of the asymptotic de-Sitter property of the spacetime, one can show that since $W(x,y)\in C^\infty(\cM\times\cM)$ is a symmetric real bi-solution to the massive wave equation, its restriction to $\cU\times\cU$ has an expansion of the form
\begin{align*}
W(x,x')=\sum\limits_{i\in\nn} c_i \psi_i(x)\psi_i(x')
\end{align*}
with $c_i=\pm 1$ and $\psi_i$ satisfying 
\begin{align*}
(\Box_g-m^2)\psi_i=b_i\, ,\,\, b_i\in C_0^\infty(\cM) 
\end{align*}
with vanishing initial data to the past of $\supp(b_i)$. In other words, the $\psi_i$ are forward solutions to the massive wave equation with source $b_i$. Moreover, the $b_i$ satisfy
\begin{align}
\label{eq:b conv}
\sum\limits_i\norm{b_i}^2_{C^m(\cM)}=C(m)<\infty\, .
\end{align}
In this way, for any $x$ in $\cU\cap\cM$, one can express
\begin{align}
\label{eq:op expansion}
\omega_1(\colon \cO^{j_1,j_2}_{\gamma_1,\gamma_2}(x)\colon)-\omega_2(\colon\cO^{j_1,j_2}_{\gamma_1,\gamma_2}(x)\colon)=\sum\limits_ic_i\left(\partial_r^{j_1}\partial_y^{\gamma_1}\psi_i(x)\right)\left(\partial_r^{j_2}\partial_y^{\gamma_2}\psi_i(x)\right)\, .
\end{align}

It then remains to estimate the derivatives of the forward solutions $\psi_i$ pointwise in $\cU\cap\cM$ in terms of some $C^m$-norm of the sources $b_i$. To do so, assume that Price's law holds in the domain $\Omega$ depicted in Figure~\ref{fig:Omega}: For a solution of $\Box_g\psi=0$ of the form $\psi=\tilde\psi+c$ with $c\in\cc$, assume $\tilde \psi\in e^{-\alpha'u}H^s(\Omega)$ with $\alpha'>0$ and $s$ sufficiently large. Here, $u(t,r)=t-r_*+U(r)$, where $U\in C^\infty([r_-,r_+);\rr)$ is chosen such that the level sets of $u$ are spacelike. The region $\Omega\subset \MII$ is defined as $\Omega=(u_2,\infty)_u\times(r_1,r_2)_r\times\ss^2$  for some $u_2=u(1-r_*(r_2),r_2)$ and $r_-<r_1<r_2<r_+$. The validity of Price's law on this domain follows from an analysis of the wave equation in a neighbourhood of $\{r_+\leq r\leq r_c\}$ \cite{PV1}.

\begin{figure}
\includegraphics[width=0.5\textwidth]{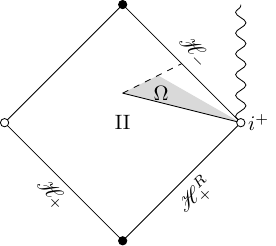}
\caption{The region $\Omega$ in which Price's law is assumed to hold. The lower boundary indicated by the black line represents $\{r=r_2\}$, while the upper boundaries are $\{r=r_1\}$ and $\{u=u_2\}$, the latter represented by the dotted line. The region enclosed by the dotted and solid lines and the inner horizon is where the estimates will be obtained, compare also \cite[Figure 2]{HK}.}
\label{fig:Omega}
\end{figure}

In \cite{HK}, it is shown in two ways that the decay provided by Price's law can be converted into the desired bounds for the limit $x\to x_0$. The first is an adaptation of the methods of \cite{HV}, working out the Fredholm property of the operator $\Box_g-i\cQ$ between suitable variable-order Sobolev spaces on a domain including the relevant piece of the inner horizon. Here, $\cQ$ is a time-translation-invariant complex absorbing operator supported in $r<r_-$, and the spacetime is modified in $r<r_-$ to include another horizon at some $r_0<r_-$. Upon restriction to $r>r_-$, this results in the required conormal regularity to obtain the desired result by Sobolev embedding. 

The second is to realise that, to first order in $r-r_-$, the wave equation can be written as a concatenation of two transport equations which can be integrated up from level sets of $(r-r_-)$ in $\Omega$, transforming decay along this surface into the desired bounds. 

This method makes use of energy estimates near $r=r_-$ to estimate $\psi$ and its derivatives in some properly weighted $L^2$-space on $(u_2,\infty)_u\times(r_-, r_2)_r\times\ss^2$ and uses the fact that $\partial_t$, $\partial_\phi$ and the Carter operator $\cC:=\tfrac{\chi^2 }{\Delta_\theta \sin^2\theta}\left(a\sin^2\theta\partial_t+\partial_\varphi\right)^2+\tfrac{1}{\sin\theta}\partial_\theta \Delta_\theta\sin\theta\partial_\theta$ commute with the wave equations to control also angular derivatives of $\psi$.

Using these classical bounds to estimate the summands in \eqref{eq:op expansion} and making use of \eqref{eq:b conv} then concludes the proof.
\qeds
\end{proof}

This result shows that the difference between expectation values of observables of the form \eqref{eq:op form} in two different Hadamard states at worst diverges like its classical analogue when $\sH_-$ is approached. As a direct consequence, the numerically computed leading divergences of different components of the stress-energy tensor in the Unruh state described above in fact constitute the universal leading divergences of the corresponding component of the stress-energy tensor. The state-dependence only enters as a subleading term as long as the spectral gap $\alpha$ is positive. 
 
\section{Summary}
\label{sec:Summary}
In this paper, we have reviewed the construction of the Unruh state for the free scalar quantum field on Kerr-de Sitter. We have seen how a novel geometric argument given in \cite{KHa} for Kerr spacetimes can be generalised to Kerr-de Sitter. This allowed us to extend the proof of the Hadamard property of the Unruh state to all subextreme Kerr-de Sitter spacetimes with sufficiently small cosmological constant, given that mode stability holds. Additionally, we have recalled the application of the Unruh state in numerical computations of the leading divergence of the stress-energy tensor at the inner horizon. The outcome of the computations, combined with a tentative, first-order analysis of the backreaction onto the spacetime geometry, suggests that the inner horizon is converted into a singularity, whose exact nature may depend on the parameters of the black hole and the quantum theory. We recalled the state independence of the leading divergence, which arises from a bound on the difference of expectation values near the horizon between different Hadamard states, and the behaviour of the classical theory. 

While the work compiled here demonstrates that there is progress in understanding the interaction of quantum fields and rotating black holes, there are still a number of open questions. First of all, even though the results in the last section certainly indicate that quantum effects can have a strong influence on the geometry at the inner horizon, this is not based on a self-consistent solution to semi-classical gravity. In fact, it does not even take into account an approximate implementation of black hole evaporation, which may already change the situation at the inner horizon significantly. 

Moreover, the scalar field is only the simplest toy model theory, and one would like to replace it with more interesting field theories such as linearised Yang-Mills theory or linearised gravity. However, while Hadamard states have been shown to exist for all globally hyperbolic spacetimes in the Yang-Mills case \cite{MS}, see also \cite{GW}, linearised gravity remains a challenging problem \cite{CMS} that is only solved in special cases \cite{GMW,GW2, BDM}.

\paragraph{\bf Data availability statement} Data sharing is not applicable to this article as no new data were created or analysed in this study.

\paragraph{\bf Conflict of interest statement} The author has no conflict to disclose.

\bibliographystyle{ieeetr}
\bibliography{Kdsbib}
\end{document}